\documentclass[conference]{IEEEtran}
\usepackage[dvips]{graphicx}
\usepackage{subfigure}
\usepackage{mdwlist}
\usepackage{amsmath}
\usepackage{amsthm}
\usepackage{algorithm}
\usepackage{algorithmic}
\usepackage{cite}
\usepackage{amssymb}
\usepackage{multirow}
\usepackage{makecell}
\usepackage{enumerate}

\newtheorem{theorem}{Theorem}
\newtheorem{lemma}{Lemma}

\begin{document}

\title{Generalized Lottery Trees: Budget-Consistent Incentive Tree Mechanisms for Crowdsourcing}
\author{\IEEEauthorblockN{Dong Zhao, Huadong Ma, and Xinna Ji}
\IEEEauthorblockA{Beijing Key Laboratory of Intelligent Telecommunications Software and Multimedia,\\Beijing University of Posts and Telecommunications, China\\
{\{dzhao, mhd\}@bupt.edu.cn, jxngood@163.com}
}
}
% make the title region
\maketitle

\begin{abstract}
Incentive mechanism design has aroused extensive attention for crowdsourcing applications in recent years. Most research assumes that participants are already in the system and aware of the existence of crowdsourcing tasks. Whereas in real life scenarios without this assumption, it is a more effective way to leverage \emph{incentive tree} mechanisms that incentivize both users' direct contributions and solicitations to other users. Although some such mechanisms have been investigated, we are the first to propose \emph{budget-consistent} incentive tree mechanisms, called \emph{generalized lottrees}, which require the total payout to all participants to be consistent with the announced budget, while guaranteeing several other desirable properties including \emph{continuing contribution incentive}, \emph{continuing solicitation incentive}, \emph{value proportional to contribution}, \emph{unprofitable solicitor bypassing}, and \emph{unprofitable sybil attack}. Moreover, we present three types of \emph{generalized lottree} mechanisms, \emph{1-Pachira}, \emph{$K$-Pachira}, and \emph{Sharing-Pachira}, which support more diversified requirements. A solid theoretical guidance to the mechanism selection is provided as well based on the Cumulative Prospect Theory. Both extensive simulations and realistic experiments with 82 users have been conducted to confirm our theoretical analysis.
\end{abstract}

\section{Introduction}
\label{sec:introduction}
In recent years, crowdsourcing has become one of the most popular distributed problem-solving model in which a crowd of undefined size is engaged to solve a complex problem through an open call \cite{chatzimilioudis2012crowdsourcing}, enabling numerous applications such as reviewing and voting items at \emph{Amazon} \cite{website:amazon} and \emph{Yelp} \cite{website:yelp}, sharing knowledge at \emph{Yahoo! Answers} \cite{website:Yahoo!Answers} and \emph{Zhihu} \cite{website:Zhihu}, creating maps at \emph{OpenStreetMap} \cite{website:OpenStreetMap}, and labeling images with the \emph{ESP} game \cite{von2004labeling}. The prevalence of these crowdsourcing applications should give the credit to various intrinsic incentives such as social, service, entertainment and ethical. On the other hand, monetary (extrinsic) incentives are leveraged by general crowdsourcing platforms like \emph{Amazon Mechanical Turk} (\emph{MTurk}) \cite{website:mturk} and \emph{Taskcn} \cite{website:taskcn} to recruit online workers for accomplishing various tasks such as data annotation, text translation, and identifying objects in a photo or video.

The proliferation of mobile sensing devices (e.g., smartphones, wearable devices, in-vehicle sensing devices) offers a new sensing paradigm as an important branch of ``crowdsourcing", which extends Web-based crowdsourcing to a larger mobile crowd, allowing to perform sensing tasks pervasively at larger scale and easier than traditional static sensor networks. This paradigm is often called as ``mobile crowd sensing" (a.k.a. ``participatory sensing", ``opportunistic sensing" with similar concepts), which has been adopted for various applications such as \emph{Sensorly} \cite{website:Sensorly}, \emph{NoiseTube} \cite{stevens2010crowdsourcing}, \emph{Common Sense} \cite{dutta2009common} for building large-scale urban sensing (network coverage, noise, and air quality) maps, \emph{Nericell} \cite{mohan2008nericell}, \emph{VTrack} \cite{thiagarajan2009vtrack} for traffic sensing, and \emph{FindingNemo} \cite{liu2014finding} for finding a lost child. Several general mobile crowdsourcing platforms such as \emph{Gigwalk} \cite{website:gigwalk}, \emph{Jana} \cite{website:jana} and \emph{Weichaishi} \cite{website:weichaishi} have recruited millions of users to participate in various mobile tasks such as conducting consumer research, launching product promotions, and creating consumer loyalty campaigns. It is always indispensable to provide proper incentives for compensating user's participating costs, including users' time, various non-negligible resources (e.g., computation, storage, and battery energy) and potential treats of leaking location privacy.

Extensive research has been conducted to design incentive mechanisms for crowdsourcing \cite{yang2012crowdsourcing,duan2012incentive,lee2010sell,jaimes2012location,zhang2015incentivize,zhang2015truthful,zhao2016budget,zhao2017frugal,guo2017taskme}. %zhang2017countermeasures
Most of existing mechanisms assume that participants are already in the system and aware of the existence of crowdsourcing tasks. However, they neglect two key facts: first, participants do not exist in the system from the beginning; second, even if there are many registered users in a crowdsourcing platform such as \emph{MTurk} and \emph{Gigwalk}, most of these potential participants are hard to timely know the existence of tasks, as many users tend to prohibit automatic task push, or turn a blind eye to tasks for saving precious time, or the platform cannot push tasks to users in right locations if users prohibit reporting their real-time GPS coordinates. On the other hand, it is a reasonable assumption that a small number of users are the first to participate in the task. For example, a user may be browsing task lists actively and decide to participate in. For another, a user may just report his location so that the platform pushes the task to him timely. Thus, a more sensible and effective method is to leverage the ``word-of-mouth" effect, namely to encourage these early participants to refer other users from their social networks such as \emph{Facebook} \cite{website:facebook}, \emph{Twitter} \cite{website:twitter} and \emph{WeChat} \cite{website:wechat}, or from their neighboring community by opportunistic networking \cite{ma2014opportunities}.

Incentive tree (a.k.a. referral tree, multi-level marketing, affiliate marketing, and direct marketing) mechanisms provide an effective way to address the aforementioned requirements. An incentive tree is a tree-structured incentive mechanism in which i) each user is rewarded for direct contributions, and in addition, ii) a user who has already participated in can make referrals, and solicit new users to also participate in and make contributions. The mechanism incentivizes solicitations by making a solicitor's reward depend on the contributions (and also on their further solicitations in a recursive manner) from such solicitees \cite{lv2016fair}. One infamous incentive tree mechanism is the \emph{Pyramid Scheme} \cite{website:pyramid}, which offers promising rewards for solicitation although being illegal in many countries. Another well-known application example is the \emph{DARPA Red Balloon Challenge}, in which an MIT team won the challenge by using a simple incentive tree mechanism \cite{pickard2011time}. However, this mechanism has a serious drawback -- not robust against sybil attacks.
%, namely offering rewards to each participant if the final balloon finder is one of its descendants in the tree structure.
%Although being efficient and winning the Red Balloon Chllenge, the incentive mechanism from the MIT team has a serious drawback -- not robust against sybil attacks. A sybil attack (a.k.a. fale-name manipulation) is that a user may generate multiple fake identifies to gain a higher reward without making extra contributions.
At present, many incentive tree mechanisms have been designed against sybil attacks \cite{emek2011mechanisms,drucker2012simpler,chen2013sybil,zhang2015sybil,lv2016fair,zhang2017robust,douceur2007lottery}, but most of them lack a budget constraint so that participants have ``unbounded reward opportunity" (a property defined in \cite{emek2011mechanisms,lv2016fair}). In fact, the crowdsourcer (i.e., crowdsourcing task organizer) often has a certain budget constraint in realistic scenarios, which also represents the mainstream incentive type in existing crowdsourcing platforms.

In this paper, we aim to design a class of budget-consistent incentive tree mechanisms satisfying six desirable properties: \emph{budget consistency} (\emph{BC}), \emph{continuing contribution incentive} (\emph{CCI}), \emph{continuing solicitation incentive} (\emph{CSI}), \emph{value proportional to contribution} (\emph{VPC}), \emph{unprofitable solicitor bypassing} (\emph{USB}), and \emph{unprofitable sybil attack} (\emph{USA}). The latter five properties are commonly considered by the existing work \cite{emek2011mechanisms,drucker2012simpler,chen2013sybil,zhang2015sybil,lv2016fair,zhang2017robust,douceur2007lottery} so that the mechanism encourages contribution, solicitation, and fair play.
Besides, we emphasize the importance of the property BC, which requires that the total payout to all participants should be consistent with the budget announced by the crowdsourcer at the time of task distribution, namely that the total payout is just equal to the budget, rather than less than the budget. Otherwise, if the total payout can be cut arbitrarily, then participants will not trust the crowdsourcer, resulting in the decline of participation enthusiasm.

To the best of our knowledge, only the work \cite{douceur2007lottery} designs a class of incentive tree mechanisms with budget constraint: \emph{lottery tree} (\emph{lottree}) mechanisms, which select one participant as the unique recipient of the payout with a probability computed by a lottery function. The \emph{Pachira lottree} is proposed to satisfy CCI, CSI, VPC, USB and USA. However, it violates BC. Moreover, it allows only one winner (for the sake of distinction, we call it \emph{1-Pachira lottree} in the rest of the paper), which is not always effective for all scenarios. In fact, it is non-trivial to adjust the \emph{1-Pachira lottree} for satisfying BC while not violating other properties, or extend it to generalized mechanisms with multiple winners, as we will elaborate later. By contrast, we propose an effective strategy to rescale the \emph{1-Pachira lottree}, and prove it satisfies all the six desirable properties. Furthermore, we design \emph{generalized Pachira lottree} mechanisms, including \emph{$K$-Pachira lottree} that allows multiple winners, and \emph{Sharing-Pachira lottree} that allows each participant to be a winner. In \emph{Sharing-Pachira lottree}, all participants proportionally share the budget based on their respective winning probabilities.
%In essence, the 1-Pachira and Sharing-Pachira lottrees are two extreme cases of the $K$-Pachira lottree when $K=1$ with one lottery drawing and $K=n$ (the number of participants) with infinite lottery drawings, respectively.

Now another key and interesting question is: which mechanism is best among the \emph{1-Pachira}, \emph{$K$-Pachira} and \emph{Sharing-Pachira lottrees}? Some recent studies \cite{reddy2010examining,musthag2011exploring,celis2013lottery,rula2014no,rokicki2014competitive,rokicki2015groupsourcing} have been conducted to compare lottery-based (a.k.a. randomized reward) and fixed payment (a.k.a. micro-payment, linear reward) mechanisms by real-world experiments. However, they lack a general and solid theoretical basis to account for their experimental results, and thus fail to provide a more persuasive guidance to the mechanism selection for different scenarios. Moreover, none of them considers incentive tree mechanisms. By contrast, we leverage the \emph{Cumulative Prospect Theory} (\emph{CPT}) \cite{tversky1992advances} to compare different \emph{generalized lottree} mechanisms by numerical analysis. This provides us an interesting and important theoretical guidance to the mechanism selection for satisfying various application requirements: \emph{If a crowdsourcer has a large budget constraint, or it only requires a small number of participants, then the \emph{Sharing-Pachira lottree} mechanism should be recommended, otherwise the \emph{1-Pachira lottree} mechanism should be recommended.}

Finally, in order to verify our theoretical analysis, we first build a social network based simulator and implement the three \emph{generalized lottree} mechanisms. Extensive simulations are conducted to confirm our theoretical analysis. Second, we investigate a typical application case: \emph{looking for lost objects}, and design an interesting experimental mobile game, \emph{Treasure Hunt}, to conduct extensive performance evaluations. 82 users register in our APP in 11 days, based on which 12 tasks are designed with different budget constraints and limits on the number of participants. Experimental results are also consistent with our theoretical analysis.

The main contributions of this paper are listed as follows:

\begin{itemize*}
\item
To the best of our knowledge, we are the first to investigate budget-consistent incentive tree mechanisms, while guaranteeing several desirable properties, CCI, CSI, VPC, USB and USA (Section II-B$\sim$C and III-A).
\item
We design \emph{generalized lottree} mechanisms in support of multiple winners, and provide theoretical guidance to the mechanism selection for satisfying different requirements by leveraging the CPT (Section III-B$\sim$E).
\item
We evaluate various mechanisms by both extensive simulations and realistic experiments to confirm our theoretical analysis, and present a typical application case (Section IV$\sim$V).
\end{itemize*} 

\section{Problem Formulation and Preliminaries}
\label{sec:problem formulation}
In this section, we first present the crowdsourcing model, the formal definition of a \emph{generalized lottree}, and the desirable properties. We next provide some preliminaries, including the \emph{1-Pachira lottree} and the CPT.

\subsection{Crowdsourcing Model}
\label{subsec:crowdsourcing model}
Suppose that there is a \emph{crowdsourcer} who requires to recruit $N$ users to participate in a crowdsourcing campaign with a budget constraint $B$. The crowdsourcer may have a specific limit on the number $N$, especially when it only requires a small number of participants. For instance, the participatory sensing data collection applications \cite{reddy2010examining}, \emph{GarbageWatch}, \textit{What's Bloomin}, and \textit{AssetLog}, only require a few motivated users to document various resource use issues at a university by taking geo-tagged photos of various resources, like outdoor waste bins, water usage of plants, bicycle racks, recycle bins, and charge stations.
Certainly, many crowdsourcing applications expect as many participants as possible, so that they often have no explicit limit on the number $N$. For instance, a large number of participants are required to build a large-scale urban sensing maps \cite{website:Sensorly,stevens2010crowdsourcing,dutta2009common} or find a lost child quickly \cite{liu2014finding}. In fact, both the number of required participants and the budget constraint have nonnegligible impacts on the effectiveness of different incentive mechanisms, as we will elaborate later.

On the other hand, users can participate in a crowdsourcing campaign and contribute to it (e.g., solving tasks, uploading sensing data, finding balloons or a lost child). Both the homogeneous and heterogeneous user models are considered \cite{zhao2017frugal}, where the former is a special case of the latter. The former can account for atom tasks where each user can complete only a single task and thus make the same contribution. For instance, \emph{Gigwalk} \cite{website:gigwalk} recruits users who are in a shopping mall for conducting consumer research, where each user can complete only one questionnaire. The latter can account for divisible tasks or crowdsourcing campaigns that users could continuously participate in, where different users may complete different numbers of tasks or participate in a campaign for different durations, resulting different contributions. For instance, Microsoft has recruited users to add panoramic images to its Bing Map results through \emph{Gigwalk} \cite{website:gigwalk}, where different users are willing to take different amounts of photos. For another, in \emph{FindingNemo} \cite{liu2014finding} different users may spend different amounts of time on looking for a lost child. Generally, the contribution of a user $u$ is denoted by $C(u)$, $C(u)\geq 0$.

Furthermore, users can also solicit new users. Such solicitations induce a tree $T$. Each user is represented as a tree node $u$, and there is a directed edge $(u,v)$ between two users $u$ and $v$ if $v$ has participated in the campaign in response to a solicitation by $u$. In other words, if $u$ participates in the campaign via a solicitation by $u$, it becomes a child-node of $v$ in $T$. The crowdsourcer is the root node $r$. The users who have participated in the campaign directly in response to the solicitation from the crowdsourcer are the child-nodes of $r$. $T_u$ denotes the subtree of $T$ rooted at node $u$. $\mathcal{E}(T)$ denotes the set of directed edges in $T$. $T$ is a weighted tree in which the weight of node $u$ is its contribution to the campaign, $C(u)$. Since the crowdsourcer has no direct contribution, we have $C(r)=0$. The total contribution of all nodes in $T$ is denoted by $C(T)=\sum_{u\in T}C(u)$.
%In addition, $\mathcal{N}(T)$ denotes the set of nodes in $T$, and $\mathcal{E}(T)$ denotes the set of edges in $T$.

\subsection{Generalized Lottree}
\label{subsec:generalized lottree}
A \emph{generalized}\footnotemark[1] \emph{lottree} is an incentive tree mechanism that leverages lotteries to probabilistically select one or multiple participants as the winner(s), and pay out a reward to each winner. One key component of a generalized lottree is a \emph{lottery} function $L(u)$ ($0\leq L(u) \leq 1$) that determines the \textit{lottery value} (i.e., winning probability) of each node $u \in T$, and satisfies $\sum_{u\in T}L(u)=1$. The lottery value of a node should depend on both the tree structure and nodes' contributions, so that both contributions and solicitations from participants are encouraged. Another key component is a \emph{reward} function $R(u)$ that determines the reward of each node $u \in T$, and satisfies $\sum_{u\in T}R(u)\leq B$. The reward of a node depends on the crowdsourcer's reward strategy, namely how many winners are allowed among $|T|-1$ participants. Specifically, three reward strategies are considered: \emph{1-lottree} with only one winner, \emph{$K$-lottree} with $K (1<K<|T|-1$) winners, and \emph{Sharing-lottree} that allows each participant to be a winner.

\footnotetext[1]{The original definition of ``lottree" allows multiple winners, but in fact such generalized cases are not considered throughout the paper \cite{douceur2007lottery}. Thus, the word ``generalized" is purposely used here to emphasize that the mechanism allows one or multiple winners.}

\subsection{Desirable Properties}
\label{subsec:desirable properties}
While the main objective of a \emph{generalized lottree} mechanism is to incentivize both contributions and solicitations under a certain budget constraint, it should also guarantee the fairness and be robust against various strategic behaviors by participants. In the following, we define the set of desirable properties that a \emph{generalized lottree} should ideally satisfy.

\textbf{Budget Consistency (BC):} A \emph{generalized lottree} satisfies \emph{BC} if the total reward to all nodes in the tree $T$ except the root node $r$ is consistent with the budget, i.e., $\sum_{u\in T\setminus \{r\}}R(u)=B$. This property has a stricter constraint than the so-called \emph{Zero Value to Root} (\emph{ZVR}) property defined in \cite{douceur2007lottery}. ZVR only requires that the reward to the root node of the tree is zero: $R(r)=0$, but allowing that the total payout is less than the budget. We argue that it is not sufficient, as it will lead to an uncommitted crowdsourcer and result in the decline of users' participation enthusiasm.

\textbf{Continuing Contribution Incentive (CCI):} A \emph{generalized lottree} satisfies \emph{CCI} if it provides nodes with increasing expected reward in response to increased contribution. Formally, given a tree $T$, if a node $u\in T$ increases its contribution, $C'(u)>C(u)$, and all other nodes $v\in T\setminus{\{u\}}$ maintain the same contribution, $C'(v)=C(v)$, then the expected reward of $u$ increases: $\mathbb{E}[R'(u)]>\mathbb{E}[R(u)]$.

\textbf{Continuing Solicitation Incentive (CSI):} A \emph{generalized lottree} satisfies \emph{CSI} if each node always has an incentive to solicit new nodes. We follow the notion of ``\emph{weak solicitation incentive} (\emph{WSI})" defined in \cite{douceur2007lottery}. Formally, if the subtree of a node $u \in T$ includes some node $p$: $p\in T_u$, but does not include some other node $q$: $q\in T\setminus T_u$, and there is a new node $n$: $C(n)>0$, which in case 1 joins the tree as a child of $p$, and in case 2 joins the tree as a child of $q$, then the reward of $u$ in case 1, denoted by $R'(u)$, is greater in expectation than that in case 2, denoted by $R''(u)$: $\mathbb{E}[R'(u)]>\mathbb{E}[R''(u)]$.

\textbf{Value Proportional to Contribution (VPC):} This property demands that the mechanism should maintain a notion of fairness among nodes, as intuitively participants expect that their rewards are proportional to their contributions. We say that a \emph{generalized lottree} satisfies $\varphi$-\emph{VPC} for some $\varphi >0$, if it ensures that the expected reward of each node $u$ is at least $\varphi$ times the relative contribution made by that node: $\mathbb{E}[R(u)]\geq \varphi C(u)/C(T)$.

\textbf{Unprofitable Solicitor Bypassing (USB):} A \emph{generalized lottree} satisfies \emph{USB} if a new node can never gain expected reward by joining the tree as a child of some node other than its solicitor. Violation of this property can result in undesirable consequences: participants will lose interest in soliciting new nodes, as new nodes tend to join the tree not as children of the nodes that solicited them. Formally, if nodes $u$ and $v$ are in the tree: $\{u,v\}\subset T$, and there is a new node $n$: $C(n)>0$, which in case 1 joins the tree as a child of $u$, and in case 2 joins the tree as a child of $v$, then the reward of $n$ in case 1, denoted by $R'(n)$, is not smaller in expectation than that in case 2, denoted by $R''(n)$: $\mathbb{E}[R'(n)]\geq \mathbb{E}[R''(n)]$, which, by symmetry, implies $\mathbb{E}[R'(n)]=\mathbb{E}[R''(n)]$.

\textbf{Unprofitable Sybil Attack (USA):} This property demands that no participant can gain lottery value by pretending to have multiple identities to join the tree as a set of \emph{Sybil} nodes instead of joining singly. Formally, the \emph{Sybil attack} is defined as follows: Given any node $u\in T$ whose parent is $p$ and who has $d\geq 0$ children $v_1, v_2, \ldots, v_d$, $u$ launches the Sybil attack by splitting itself into multiple replicas (i.e., Sybil nodes), denoted by $u_1, u_2, \ldots, u_s$ ($s>1$), $\sum_{i=1}^{s}{C(u_i)=C(u)}$; each Sybil node $u_i$ can only be a child of $p$, or a child of one of the other Sybil nodes; each node $v_j$ is a child of one Sybil node. The total expected reward for $u$ from this Sybil attack is $\sum_{i=1}^{s}{\mathbb{E}[R(u_i)]}$. We say that a \emph{generalized lottree} satisfies \emph{USA} if for any node $u\in T$, it cannot gain expected reward by any Sybil attack without making extra contributions: $\mathbb{E}[R(u)]\geq \sum_{i=1}^{s}{\mathbb{E}[R(u_i)]}$.

\subsection{1-Pachira Lottree}
\label{subsec:1-Pachira}
The \emph{1-Pachira lottree} has been proven to satisfy CCI, CSI, VPC, USB and USA \cite{douceur2007lottery}. In principle, the \emph{1-Pachira lottree} can be defined using any function $\pi$ that satisfies the following properties:
\begin{enumerate}[(i)]\setlength{\itemsep}{0.1cm}
\item $\pi(0)=0$, $\pi(1)=1$;
\item $\forall c\in [0,1]$: $\frac{d\pi(c)}{dc}\geq \beta\quad$ (minimum slope of $\beta$);
\item $\forall c\in [0,1]$: $\frac{d^2\pi(c)}{dc^2}>0\quad$ (strictly convex).
\end{enumerate}

In this paper, we follow a particularly convenient and intuitive function with the above properties:
\begin{equation}\label{eq-pi}
  \pi(c)=\beta c+(1-\beta)c^{1+\delta},
\end{equation}
where $\beta$ and $\delta$ are two input parameters that tradeoff solicitation incentive against fairness. Then for each node $u\in T$, a \emph{weight} is computed as the function $\pi$ applied to its proportional contribution: $W(u)=\pi(C(u)/C(T))$. Besides, the weight of a subtree $T_u$ is defined as
\begin{equation}\label{eq-weight}
W(T_u)=\pi \left(\frac{C(T_u)}{C(T)}\right).
\end{equation}
Specially, for any leaf node $u$, it holds that $W(u)=W(T_u)=\pi(C(u)/C(T))$. Finally, the \emph{1-Pachira lottree} determines the lottery value of each node $u\in T$ as the weight of the subtree rotted at $u$ minus the weights of all $u$'s child subtrees:
\begin{equation}\label{eq-lottery}
  L(u)=W(T_u)-\sum_{(u,v)\in \mathcal{E}(T)} W(T_v).
\end{equation}
Only one node obtains all the reward with probability $L(u)$.

\subsection{Cumulative Prospect Theory (CPT)}
\label{subsec:CPT}
An effective lottree mechanism should consider how people perceive the payout for different reward strategies based on the cognitive psychology of lottery gambling \cite{rogers1998cognitive}. For this purpose, the \emph{Prospect Theory}, a generally accepted economic model proposed by Kahneman and Tversky \cite{kahneman1979prospect}, can adequately describe how individuals evaluate losses and gains in lotteries instead of the expected utility theory. Furthermore, the \emph{Cumulative Prospect Theory} (\emph{CPT}) extends it to uncertain as well to risky prospects with any number of outcomes, which also confirms a distinctive fourfold pattern of risk attitudes: risk aversion for gains and risk seeking for losses of high probability; risk seeking for gains and risk aversion for losses of low probability \cite{tversky1992advances}. Note that only the gain case is considered for lottrees. Specifically, for a single outcome of the gain $x\geq 0$ with a probability $\mathfrak{p}$, the \emph{value} function and \emph{weighting} function are respectively defined based on a nonlinear transformation as follows:
\begin{equation}\label{eq-value}
  \nu(x)=x^\alpha,
\end{equation}
\begin{equation}\label{eq-weighting}
  \omega^+(\mathfrak{p})=\frac{\mathfrak{p}^\gamma}{(\mathfrak{p}^\gamma+(1-\mathfrak{p})^\gamma)^{1/\gamma}}.
\end{equation}
The \emph{cumulative prospect value (CPV)} is then computed as individuals' perceived gain:
\begin{equation}\label{eq-cpv}
  CPV(x, \mathfrak{p})=\nu(x)\omega^+(\mathfrak{p}).
\end{equation}

Furthermore, if there are a series of possible outcomes with gain-probability pairs $(x_i, \mathfrak{p}_i)$, $1\leq i \leq m$, then the \emph{cumulative decision weight} is defined by:
\begin{equation}\label{eq-cumulative-weight}
  \tau_i^+=\omega^+(\mathfrak{p}_i+\cdots+\mathfrak{p}_m)-\omega^+(\mathfrak{p}_{i+1}+\cdots+\mathfrak{p}_m), 1\leq i < m,
\end{equation}
\begin{equation}\label{eq-cumulative-weight-m}
  \tau_m^+=\omega^+(\mathfrak{p}_m).
\end{equation}
The CPV is then computed as:
\begin{equation}\label{eq-cpv-2}
  CPV((x_i, \mathfrak{p}_i), 1\leq i \leq m)=\sum_{i=1}^m \tau_i^+\nu(x_i).
\end{equation} 

\section{Generalized Pachira Lottree}
\label{sec:generalized pachira tree}
In this section, we first present how to rescale the \emph{1-Pachira lottree} for guaranteeing all the desirable properties. Note that, since each node's reward depends only on its own lottery value for the \emph{1-Pachira lottree}, it is only required to consider the lottery value for analyzing various properties. We then extend it to \emph{K-Pachira lottree} and \emph{Sharing-Pachira lottree}, respectively. Finally, we analyze how to select mechanisms based on the CPT, and give an important theoretical guideline.

\subsection{Rescaling 1-Pachira Lottree}
\label{subsec:rescaling}
The \emph{1-Pachira lottree} does not satisfy ZVR, because the root node can obtain the reward with probability $L(r)>0$. It also violates BC, as ZVR is a necessary but not sufficient condition for BC. It is a straightforward strategy to rescale the lottree for satisfying ZVR by distributing the root's lottery value to the other nodes. Interestingly, however, it is non-trivial for a rescaling strategy to ensure the desirable properties, especially for BC, USB and USA. For instance, a rescaling strategy is proposed by \cite{douceur2007lottery} to distribute the root's lottery value among the other nodes in proportion to their lottery values (Fig. \ref{fig_rescaling2}), but it violates USB. Although another rescaling strategy is further used to satisfy USB, it still violates BC. Intuitively, a crowdsourcer may tend to use the other two rescaling strategies: one is to distribute the root's lottery value to nodes who are at a lower level of the tree, e.g., the second-level nodes who join the tree directly in response to the solicitation from the crowdsourcer (Fig. \ref{fig_rescaling3}); the other is to distribute the root's lottery value to nodes who join the tree earlier, e.g., the first two nodes who join the tree regardless of the tree structure (Fig. \ref{fig_rescaling4}). Both the two strategies encourage users to participate in the campaign as soon as possible, which is also a good property for an incentive mechanism. In the following, we define two more general rescaling strategies.
\begin{figure*}[!t]
  \leftline{
  \subfigure[Before rescaling]{
    \label{fig_rescaling1}
    \includegraphics[height=0.9in]{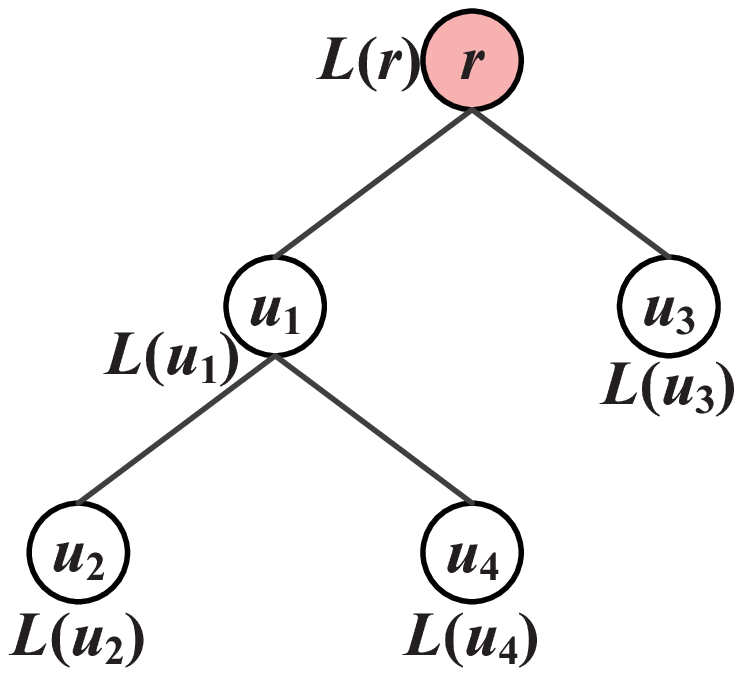}\vspace{-15pt}}
  \subfigure[After rescaling of \cite{douceur2007lottery}]{
    \label{fig_rescaling2}
    \includegraphics[height=0.9in]{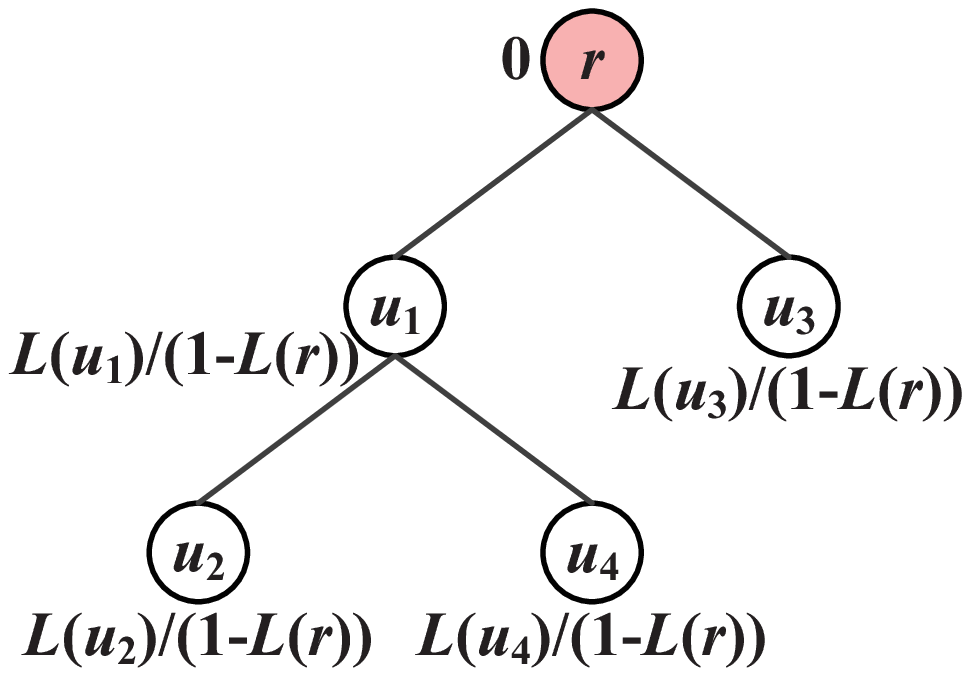}\vspace{-15pt}}
  \subfigure[After structure-dependent rescaling ($w_1+w_3=1$)]{
    \label{fig_rescaling3}
    \includegraphics[height=0.9in]{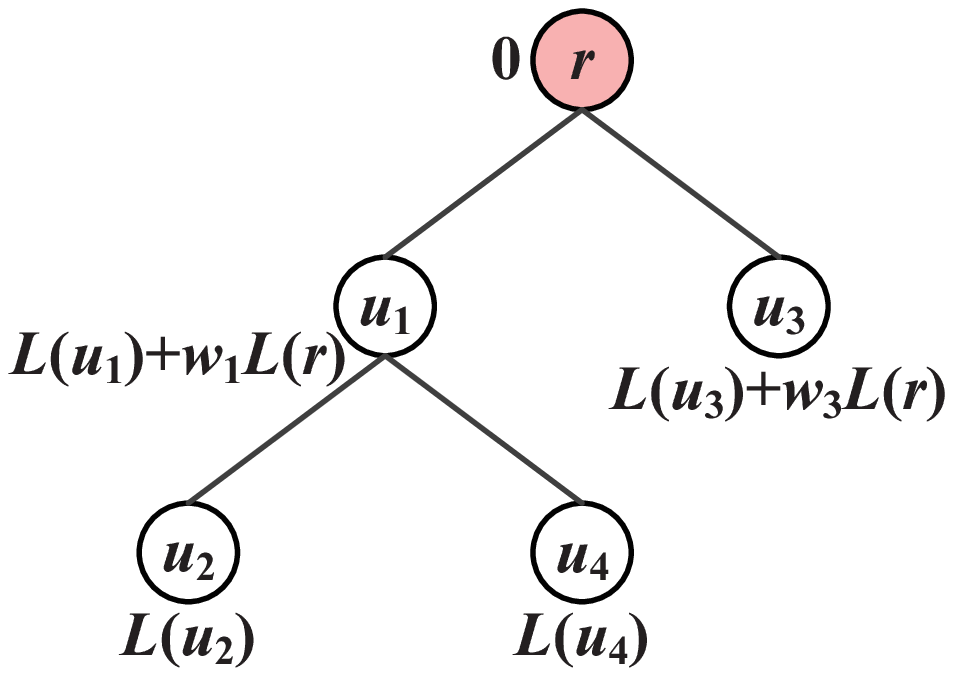}\vspace{-15pt}}
  \subfigure[After time-dependent rescaling ($w_1+w_2=1$)]{
    \label{fig_rescaling4}
    \includegraphics[height=0.9in]{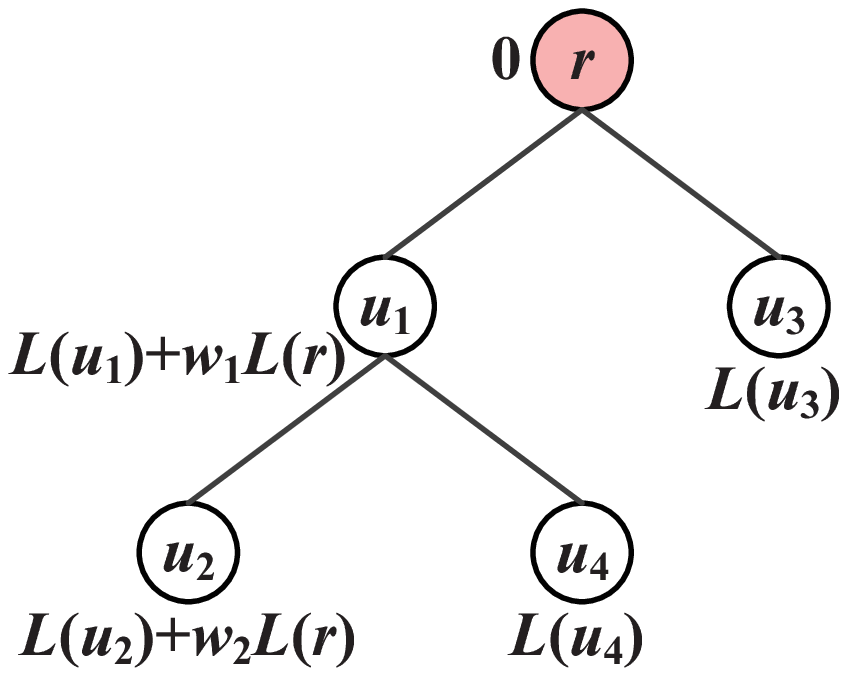}\vspace{-15pt}}
  \subfigure[After first-is-root rescaling]{
    \label{fig_rescaling5}
    \includegraphics[height=0.9in]{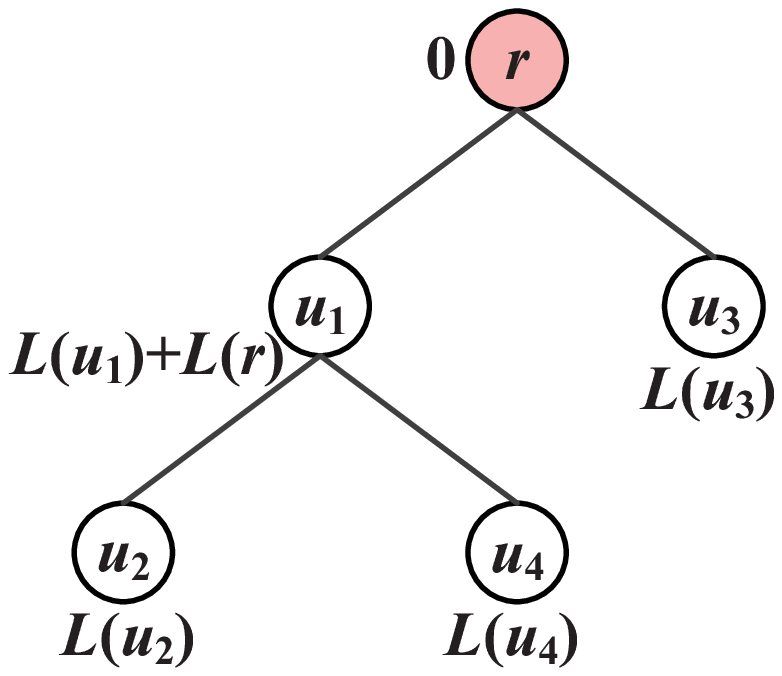}\vspace{-15pt}}
  \subfigure[Equivalent of the lottree in (e)]{
    \label{fig_rescaling6}
    \includegraphics[height=0.6in]{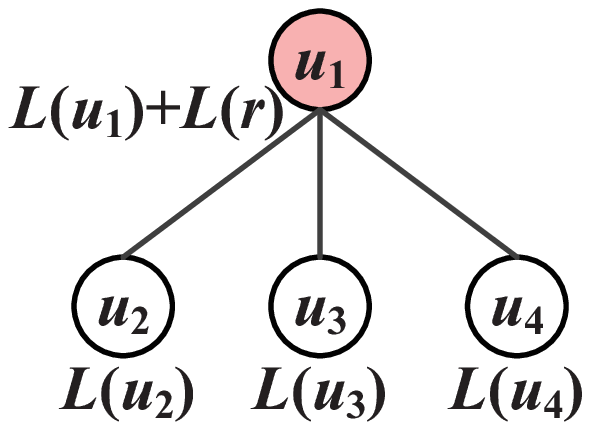}\vspace{-15pt}}
  }
  \caption{Illustration of different rescaling strategies for a 1-Pachira lottree with a crowdsourcer $r$ and four participants $u_1\backsim u_4$, where $u_i$ denotes the $i$-th user based on the time of joining the tree, $L(u_i)$ and $L(r)$ denote the original lottery values according to the original tree structure and lottery function Eq. (\ref{eq-lottery}), and the value beside each node denotes the respective lottery value before or after rescaling.}
  \label{fig_rescaling} %% label for entire figure
  \vspace{-10pt}
\end{figure*}

\textbf{Structure-dependent Rescaling:} Given a lottree with lottery value $L(u)$ for each non-root node $u$ and $L(r)$ for the root node $r$, a rescaling strategy is \emph{structure-dependent} if it distributes $r$'s lottery value to part of nodes $\{v_i\}$ who are at the specified levels or locations based on the tree structure with proportion $w_i>0$, $1\leq i \leq l$, so that $v_i$'s lottery value is rescaled as $L'(v_i)=L(v_i)+w_i L(r)$, $\sum_{i=1}^l w_i =1$, $r$'s lottery value is rescaled as $L'(r)=0$, and other nodes' lottery values remain the same.

\textbf{Time-dependent Rescaling:} Given a lottree with lottery value $L(u)$ for each non-root node $u$ and $L(r)$ for the root node $r$, a rescaling strategy is \emph{time-dependent} if it distributes $r$'s lottery value to part of nodes $\{v_i\}$ who are in the specified time orders joining the tree with proportion $w_i>0$, $1\leq i \leq l$, so that $v_i$'s lottery value is rescaled as $L'(v_i)=L(v_i)+w_i L(r)$, $\sum_{i=1}^l w_i =1$, $r$'s lottery value is rescaled as $L'(r)=0$, and other nodes' lottery values remain the same.

Besides, we define the \emph{first-is-root} rescaling strategy as the special case of both \emph{structure-dependent} and \emph{time-dependent} rescaling strategies, as illustrated in Fig. \ref{fig_rescaling5} and Fig. \ref{fig_rescaling6}.

\textbf{First-is-root Rescaling:} Given a lottree with lottery value $L(u)$ for each non-root node $u$ and $L(r)$ for the root node $r$, a \emph{first-is-root} rescaling strategy rescales the lottery value of node $u_1$, who is the first to join the tree, as $L'(u_1)=L(u_1)+L(r)$, $r$'s lottery value as $L'(r)=0$, and leaves other nodes' lottery values unchanged.

We next analyze the above rescaling strategies one by one.
\begin{theorem}
\label{theorem:structure-dependent}
Any structure-dependent rescaling strategy violates USB except the first-is-root rescaling.
\end{theorem}

\begin{proof}
Let us first consider any structure-dependent rescaling strategy except the first-is-root rescaling, where some nodes will tend to bypass their solicitors strategically to become the nodes who are at the specified levels or locations for gaining a higher lottery value. It can be proved by the following example: For a lottree illustrated in Fig. \ref{fig_rescaling1}, if user $u_4$ with lottery value $L(u_4)$ bypasses its solicitor $u_1$ to become the child of $r$, then its lottery value $L'(u_4)$ remains the same in the case without structure-dependent rescaling as illustrated in Fig. \ref{fig_USB1}: $L'(u_4)=W(u_4)=\pi(C(u_4)/C(T))=L(u_4)$, but its lottery value $L''(u_4)$ can be increased after the structure-dependent rescaling as illustrated in Fig. \ref{fig_USB2}: $L''(u_4)=L(u_4)+w_4L_{SB}(r)>L(u_4)$. This means that USB is violated after a structure-dependent rescaling, as $u_4$ can gain a higher lottery value by bypassing its solicitor $u_1$.

Then we consider the first-is-root rescaling strategy.
In this special case, no node can bypass its solicitor to become the first node.
It implies that no node can gain a higher lottery value by bypassing its solicitor.
Thus, the first-is-root rescaling strategy satisfies USB.
\end{proof}
\begin{figure}[!t]
  \centerline{
  \subfigure[Before rescaling]{
    \label{fig_USB1}
    \includegraphics[height=0.9in]{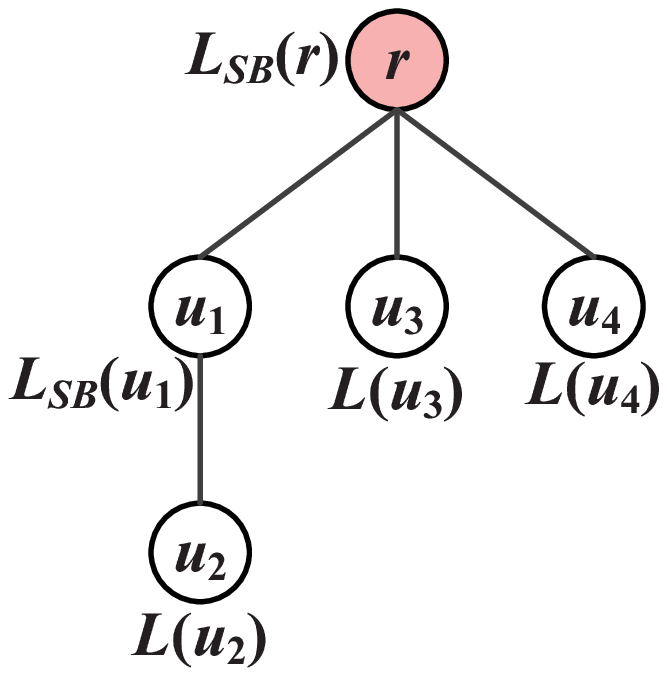}\hspace{15pt}\vspace{-15pt}}
  \subfigure[After structure-dependent rescaling ($w_1+w_3+w_4=1$)]{
    \label{fig_USB2}
    \includegraphics[height=0.9in]{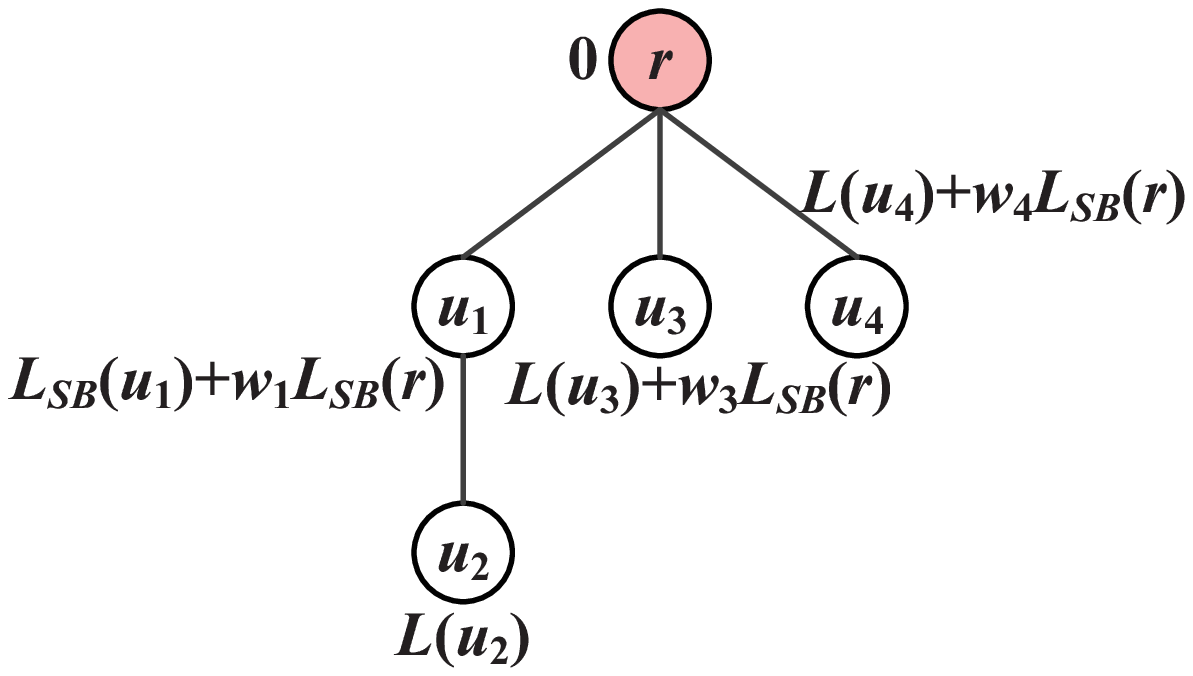}\vspace{-15pt}}
  }
  \caption{Illustration of the impact from the structure-dependent rescaling strategy on \emph{USB} when $u_4$ bypasses its solicitor in Fig. \ref{fig_rescaling1}, where $L_{SB}(*)$ denotes a new lottery value caused by the \emph{solicitor bypassing} behavior.}
  \label{fig_USB} %% label for entire figure
  \vspace{-10pt}
\end{figure}

\begin{theorem}
\label{theorem:time-dependent}
Any time-dependent rescaling strategy violates USA except the first-is-root rescaling.
\end{theorem}
\begin{figure}[!t]
  \centerline{
  \subfigure[Before rescaling]{
    \label{fig_USA1}
    \includegraphics[height=0.9in]{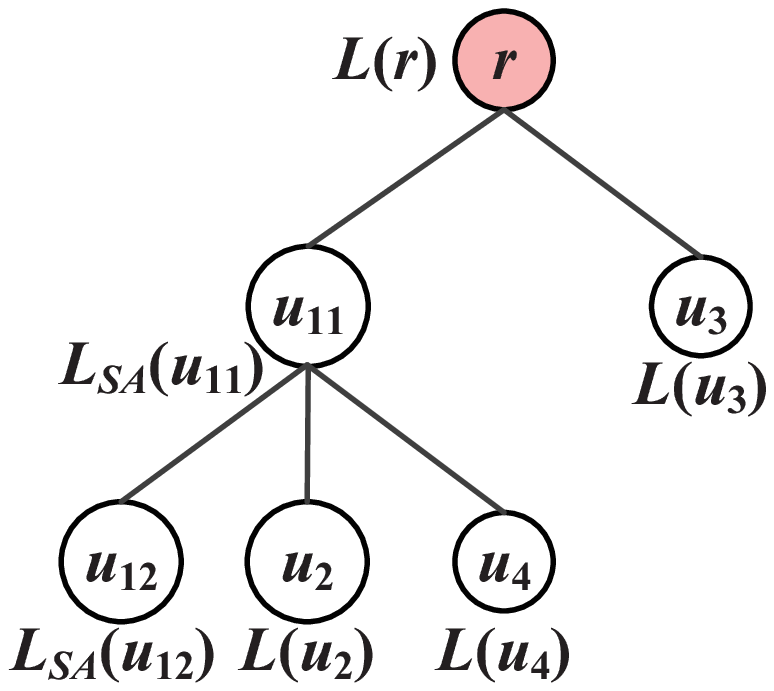}\hspace{15pt}\vspace{-15pt}}
  \subfigure[After time-dependent rescaling ($w_1+w_2=1$)]{
    \label{fig_USA2}
    \includegraphics[height=0.9in]{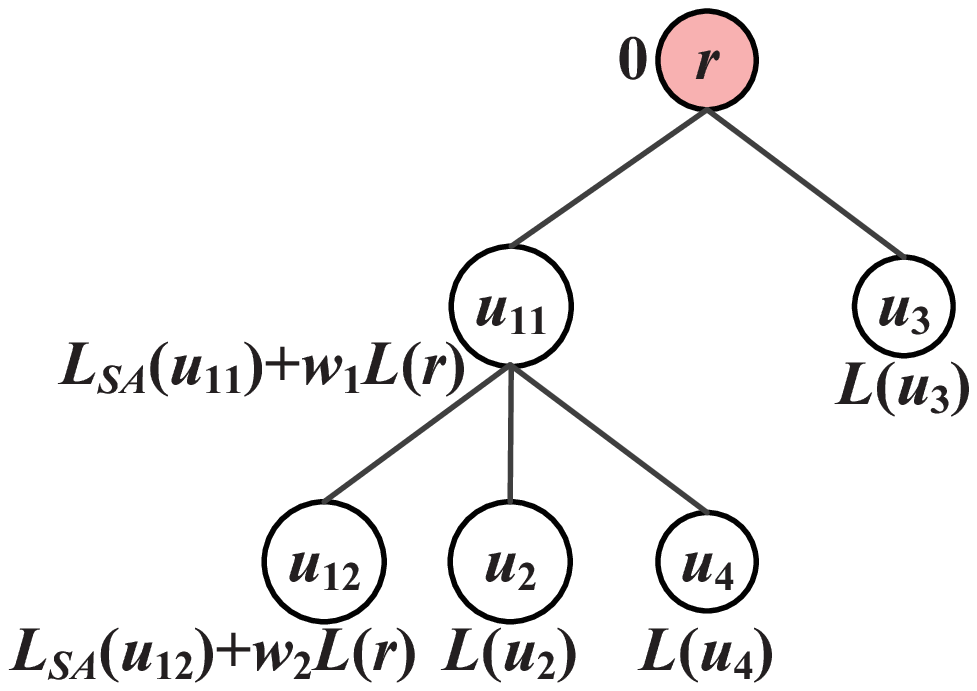}\vspace{-15pt}}
  }
  \caption{Illustration of the impact from the time-dependent rescaling strategy on USA when $u_1$ launches a Sybil attack by splitting itself into two nodes $u_{11}$ and $u_{12}$ in Fig. \ref{fig_rescaling1}, where $L_{SA}(*)$ denotes a new lottery value caused by the Sybil attack.}
  \label{fig_USA} %% label for entire figure
  \vspace{-10pt}
\end{figure}
\begin{proof}
Let us first consider any time-dependent rescaling strategy except the first-is-root rescaling, where some nodes will tend to launch a Sybil attack by having multiple identities in the specified time orders joining the tree for gaining a higher lottery value. It can be proved by the following example: For a lottree $T$ illustrated in Fig. \ref{fig_rescaling1}, if user $u_1$ with lottery value $L(u_1)$ splits itself into two nodes $u_{11}$ and $u_{12}$, which are the first two nodes joining the tree, resulting in a new lottree $T'$ illustrated in Fig. \ref{fig_USA1}, then $u_1$ cannot gain lottery value in the case without time-dependent rescaling:
\begin{align}\label{eq-time}
&L_{SA}(u_{11})+L_{SA}(u_{12})\nonumber\\
=&W(T'_{u_{11}})-[W(u_{12})+W(u_2)+W(u_4)]+W(u_{12})\nonumber\\
=&W(T'_{u_{11}})-[W(u_2)+W(u_4)]\nonumber\\
=&\pi\left(\frac{C(u_{11})\!+\!C(u_{12})\!+\!C(u_2)\!+\!C(u_4)}{C(T)}\right)\!-\![W(u_2)\!+\!W(u_4)]\nonumber\\
=&\pi\left(\frac{C(u_1)+C(u_2)+C(u_4)}{C(T)}\right)-[W(u_2)+W(u_4)]\nonumber\\
=&W(T_{u_1})-[W(u_2)+W(u_4)]\nonumber\\
=&L(u_1),
\end{align}
which also implies that the lottery value of $r$, $L(r)$, remains the same.
However, a time-dependent rescaling will result in a new lottree $T''$ illustrated in Fig. \ref{fig_USA2}, where $u_1$'s lottery value $L''(u_1)$ will gain as follows:
\begin{align}\label{eq-time2}
L''(u_1)&=L_{SA}(u_{11})+w_1L(r)+L_{SA}(u_{12})+w_2L(r)\nonumber\\
&=L_{SA}(u_{11})+L_{SA}(u_{12})+L(r)\nonumber\\
&=L(u_1)+L(r)\nonumber\\
&>L(u_1).
\end{align}
This means that USA is violated after a time-dependent rescaling, as $u_1$ can gain a higher lottery value by launching a Sybil attack.

Then we consider the first-is-root rescaling strategy.
For any node in the lottree except the first node, it is impossible to gain a higher lottery value by launching a Sybil attack, as none of its sybil nodes may become the first node to occupy the original lottery value of the root node.
Now we consider the first node $u_1$.
According to the definition of the first-is-root rescaling, if $u_1$ does not launch any Sybil attack, its lottery value will be:
\begin{equation}\label{eq-first-is-root1}
  L'(u_1)=L(u_1)+L(r)=1-\sum_{u_i \in {T\setminus\{r,u_1\}}}{L(u_i)}.
\end{equation}
If $u_1$ launches a Sybil attack by splitting itself into $s$ replicas $u_{11}, u_{12}, \ldots, u_{1s}$ ($s>1$), $\sum_{i=1}^{s}{C(u_{1i})=C(u_1)}$, then the original lottery value of the root node is only distributed to the first Sybil node $u_{11}$ regardless of how $u_1$ organize the structure of Sybil nodes, and its lottery value will be:
\begin{align}\label{eq-first-is-root2}
  \sum_{i=1}^s{L''(u_{1i})}&=L_{SA}(u_{11})+L(r)+\sum_{i=2}^s{L_{SA}(u_{1i})}\nonumber\\
  &=1-\sum_{u_i \in {T\setminus\{r,u_1\}}}{L(u_i)}.
\end{align}
This implies that the first node is also impossible to gain a higher lottery value by launching a Sybil attack.
Thus, the first-is-root rescaling strategy satisfies USA.
\end{proof}

According to the definition of the first-is-root rescaling strategy, it is straightforward to satisfy BC. It is also easy to infer that the first-is-root rescaling strategy satisfies CCI and VPC, as the lottery value gets higher for the first node $u_1$ and remains unchanged for any other node $v\in {T\setminus\{r,u_1\}}$. One may think that in the first-is-root rescaling strategy the first node will have no incentive to solicit new nodes as it will become the root node and all other nodes are its descendants regardless of whether it makes referrals. Considering this, one may think the first-is-root rescaling strategy violates CSI. In fact, however, it can be proved that the first node still has an incentive to solicit new nodes due to the competitive effect.
%\begin{lemma}
%\label{lemma:BC}
%The first-is-root rescaling strategy satisfies BC.
%\end{lemma}
\begin{lemma}
\label{lemma:CSI}
The first-is-root rescaling strategy satisfies CSI.
\end{lemma}
\begin{proof}
Since the first-is-root rescaling strategy does not change the lottery value of any node except the root node and the first node, it is only required to prove that the first node has an incentive to solicit new nodes. Let us consider an example illustrated in Fig. \ref{fig_example}: $u_1$ and $u_2$ are the first two nodes. According to the first-is-root rescaling strategy, $u_1$ is rescaled as the root node, and $u_2$ is the child of $u_1$, as illustrated in Fig. \ref{fig_example1}. Now assume that there is a new node $u_3$: $C(u_3)>0$, which in case 1 joins the tree in response of $u_1$'s solicitation (Fig. \ref{fig_example2}), and in case 2 joins the tree in response of $u_2$'s solicitation (Fig. \ref{fig_example3}). The lottery values of $u_1$ in the two cases are respectively as follows:
\begin{equation}\label{eq-CSI1}
  L'(u_1)=1-\left[\pi\left(\frac{C(u_2)}{C(T)}\right)+\pi\left(\frac{C(u_3)}{C(T)}\right)\right].
\end{equation}
\begin{equation}\label{eq-CSI1}
  L''(u_1)=1-\pi\left(\frac{C(u_2)+C(u_3)}{C(T)}\right).
\end{equation}
Due to the strict convexity of the function $\pi$, the following inequality holds:
\begin{equation}\label{eq-CSI3}
\pi\left(\frac{C(u_2)}{C(T)}\right)+\pi\left(\frac{C(u_3)}{C(T)}\right)<\pi\left(\frac{C(u_2)+C(u_3)}{C(T)}\right),
\end{equation}
which implies that $L'(u_1) > L''(u_1)$. Thus, CSI is satisfied.
\end{proof}
\begin{figure}[!t]
  \centerline{
  \subfigure[]{
    \label{fig_example1}
    \includegraphics[height=0.7in]{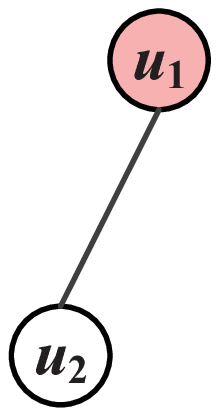}\hspace{20pt}\vspace{-15pt}}
  \subfigure[]{
    \label{fig_example2}
    \includegraphics[height=0.7in]{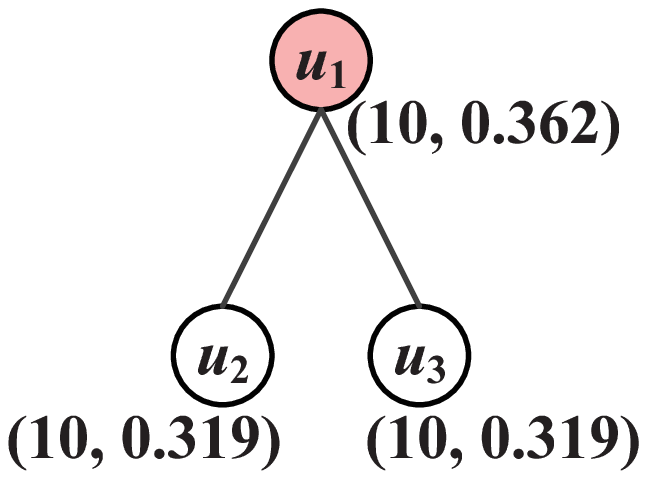}\hspace{20pt}\vspace{-15pt}}
  \subfigure[]{
    \label{fig_example3}
    \includegraphics[height=0.7in]{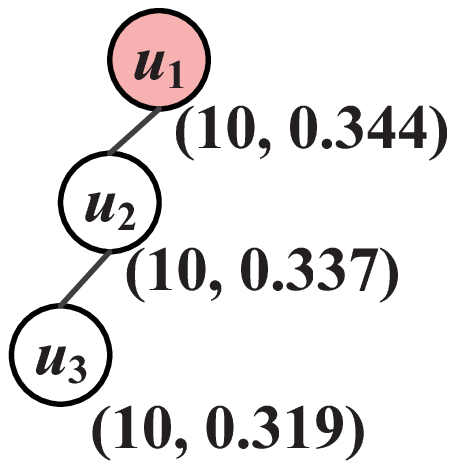}\vspace{-15pt}}
  }
  \caption{Illustration of a lottree using the first-is-root rescaling strategy. (a) denotes the initial lottree; (b) and (c) denote two cases where new node $u_3$ joins the tree, and the 2-tuple beside each node $u_i$ denotes the respective contribution and lottery value: $(C(u_i),L(u_i))$.}
  \label{fig_example} %% label for entire figure
  \vspace{-10pt}
\end{figure}

The aforementioned analysis together prove the following theorem.
\begin{theorem}
\label{theorem:first-is-all}
The first-is-root rescaling strategy satisfies all desirable properties, including BC, CCI, CSI, VPC, USB and USA.
\end{theorem}

Besides, the first-is-root rescaling strategy has an additional advantage: users will compete to be the first participant for winning the extra lottery value, which is benefit for recruiting the first batch of users as soon as possible.

\subsection{$K$-Pachira Lottree}
\label{subsec:k-pachira}
One basic problem for extending the \emph{1-Pachira lottree} to $K$-\emph{Pachira} \emph{lottree} is how to select $K$ winners based on users' lottery values. Generally, there are four potential strategies:

\textbf{Strategy \emph{A}:} $K$ different winners are selected in $K$ rounds. After each round, the selected winner is excluded from the candidate set, and the next winner is selected from the rest of nodes based on their respective lottery values.

\textbf{Strategy \emph{B}:} All nodes are sorted according to their decreasing lottery values, and then the top $K$ nodes are selected as winners.

\textbf{Strategy \emph{C}:} $K$ winners are selected in $K$ rounds. In each round, a winner is selected from all nodes based on their respective lottery values, and it is never excluded from the candidate set. It means a node may be selected as winners multiple times.

\textbf{Strategy \emph{D}:} All nodes are allocated virtual lottery tickets proportionately based on their respective lottery values, and then $K$ tickets are drew randomly in one round to determine their owners as winners.

It is interesting to see that both Strategy \emph{A} and \emph{B} violate USB, while both Strategy \emph{C} and \emph{D} maintain all desirable properties of the \emph{1-Pachira lottree} with first-is-root rescaling. We first analyze Strategy \emph{A} by an example illustrated in Fig. \ref{fig_example}: Assume that both two nodes $u_1$ and $u_2$ in Fig. \ref{fig_example1} solicit a new node $u_3$, and then $u_3$ in case 1 joins the tree in response of $u_1$'s solicitation (Fig. \ref{fig_example2}), and in case 2 joins the tree in response of $u_2$'s solicitation (Fig. \ref{fig_example3}). Meanwhile, assume that three nodes have the same contribution, 10, and the \emph{$K$-Pachira lottree} is adopted with $K=2$. We can get different lottery values for case 1 and 2, as illustrated in Fig. \ref{fig_example2} and \ref{fig_example3}. Furthermore, we can compute the probability that $u_3$ becomes one of the two winners:
\begin{equation}\label{eq-example}
P(u_3)=L(u_3)+\frac{L(u_1)L(u_3)}{L(u_1)+L(u_3)}+\frac{L(u_2)L(u_3)}{L(u_2)+L(u_3)},
\end{equation}
which equals to 0.649 and 0.648 for case 1 and case 2, respectively. Obviously, this violates USB.
% The essential reason is that one node's final winning probability depends on other nodes' lottery values.
Specifically, a new node tends to become the child of a solicitor with higher lottery value, so that it has a better chance to win in the next round after the node with higher lottery value is selected and excluded from the candidate set.

Strategy \emph{B} is a competitive strategy in essence. It is not difficult to infer that a new node tends to become the child of a solicitor with lottery value higher than itself so as to maintain its competitive advantage over other nodes with lower lottery values. Thus, Strategy \emph{B} also violates USB.

More generally, in order to satisfy USB, each node's final winning probability should be independent of other nodes' lottery values. Both strategies \emph{C} and \emph{D} follow this principle, and thus satisfy USB. In essence, Strategy \emph{C} and \emph{D} are equivalent to the sampling with replacement and the sampling without replacement in the probability theory, respectively. Specifically, each node $u$ has the same winning probability $L(u)$ in each round for Strategy \emph{C}, and the same node has a slightly higher probability of winning at least once for Strategy \emph{D}. It is also not difficult to see that both Strategy \emph{C} and \emph{D} maintain other desirable properties.

After determining $K$ winners, another basic problem for extending the \emph{1-Pachira lottree} to \emph{$K$-Pachira lottree} is how to allocate rewards to these winners. It should follow a similar principle, namely that each node's final reward should be independent of other nodes' lottery values. It is a good choice to allocate the total reward $B$ equally to $K$ winners, where each node has the same expected reward as that by using the \emph{1-Pachira lottree} mechanism. In the rest of paper, when referring to the \emph{$K$-Pachira lottree}, we use Strategy \emph{C} together with the reward equipartition strategy for convenience.

\subsection{Sharing-Pachira Lottree}
\label{subsec:sharing-pachira}
In essence, the \emph{Sharing-Pachira lottree} is equivalent to one extreme case of the \emph{$K$-Pachira lottree} with infinite lottery drawings. Under this case, all nodes will proportionally share the budget based on their respective lottery values. In other words, each node $u$ will share a reward:
\begin{equation}\label{eq-sharing}
R(u)=B*L(u).
\end{equation}

It is easy to know that the \emph{Sharing-Pachira lottree} maintains all desirable properties as each node's reward is independent of other nodes' lottery values.

\subsection{CPT-based Mechanism Selection}
\label{subsec:mechanism selection}
\begin{figure}[!t]
  \centering{
  \subfigure[$B=100$]{
    \label{fig_CPT1}
    \includegraphics[width=3.0in]{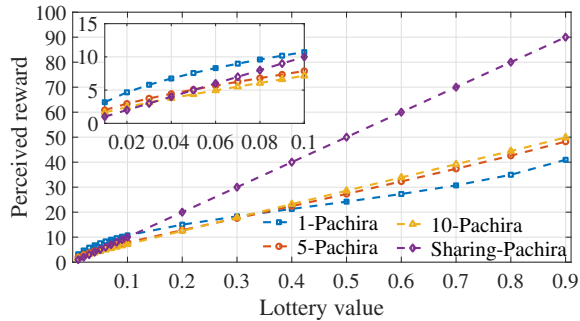}\vspace{-15pt}}
  \subfigure[$B=1000$]{
    \label{fig_CPT2}
    \includegraphics[width=3.0in]{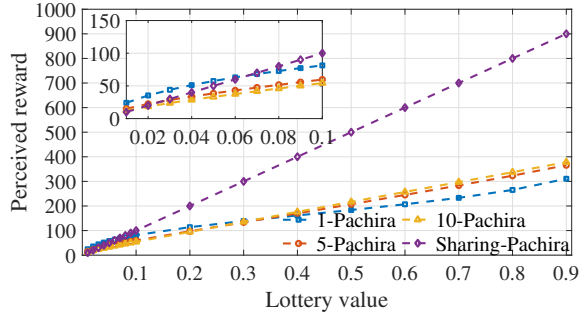}\vspace{-15pt}}
  }
  \caption{A user's perceived reward with various lottery values.}
  \label{fig_CPT}
  \vspace{-10pt}
\end{figure}

Since each user has an uncertain reward before the end of a crowdsourcing campaign for both \emph{1-Pachira} and \emph{$K$-Pachira lottrees}, it is important how a user perceives the payout, which decides whether the user is willing to make contributions and solicitations.
As introduced in Sec. \ref{subsec:CPT}, we can leverage CPT to analyze how users perceive the payout for different reward mechanisms. For the \emph{1-Pachira lottree}, the perceived reward for each user $u$ can be computed according to Eqs. (\ref{eq-value})-(\ref{eq-cpv}) by using the total reward $B$ as the gain $x$ and the lottery value $L(u)$ as the probability $\mathfrak{p}$. For the \emph{$K$-Pachira lottree}, each user $u$ has $K$ possible outcomes with gain-probability pairs:
\begin{equation}\label{gain-probability}
\left(\frac{B\cdot i}{K}, C_K^i (L(u))^i (1-L(u))^{K-i}\right), 1\leq i \leq K,
\end{equation}
and then the perceived reward for user $u$ can be computed according to Eqs. (\ref{eq-cumulative-weight})-(\ref{eq-cpv-2}). For the \emph{Sharing-Pachira lottree}, each user $u$ has a certain reward as shown in Eq. (\ref{eq-sharing}). We compare a user's perceived reward with various lottery values and two different budgets ($B=100, 1000$) for 4 mechanisms (\emph{1-Pachira}, \emph{5-Pachira}, \emph{10-Pachira}, and \emph{Sharing-Pachira}), as shown in Fig. \ref{fig_CPT}. Three interesting phenomenons can be observed:

i) When the budget stays the same, a user with a lower lottery value perceives the largest payout by the \emph{1-Pachira lottree} among all mechanisms, and a user with a higher lottery value perceives the largest payout by the \emph{Sharing-Pachira lottree}, whereas the \emph{$K$-Pachira lottree} always stays the middle level regardless of the lottery value.

ii) When the budget stays the same, a distinct critical lottery value exists, below which one may prefer the \emph{1-Pachira lottree} to the \emph{Sharing-Pachira lottree}, and above which one may prefer the \emph{Sharing-Pachira lottree} to the \emph{1-Pachira lottree}.

iii) As the budget increases, the critical lottery value will become larger.

In essence, the above observations are consistent with the CPT that one tends to risk seeking for higher gains of low probability, and risk aversion for lower gains of high probability. This provides us an interesting and important theoretical guidance to the mechanism selection for satisfying different application requirements as follows.

\textbf{Guidance to Mechanism Selection:} \emph{If a crowdsourcer has a large budget constraint, or it only requires a small number of participants, then the Sharing-Pachira lottree mechanism should be recommended, otherwise the 1-Pachira mechanism should be recommended.}

%\subsection{Mechanism Implementation}
%\label{subsec:implementation}

\section{Performance Evaluation By Simulations}
\label{sec:simulations}
To evaluate the performance of different lottree mechanisms under various scenarios, we build a simulator and conduct extensive simulations. Moreover, the impacts of the budget constraint and the number of required participants are investigated. In this section, we present the simulation framework, parameter settings, and simulation results.

\subsection{Simulation Framework and Parameter Settings}
\label{subsec:simulation framework}
We build a simulator based on the following four steps:

\textbf{Step 1):} The crowdsourcer pushes the crowdsourcing campaign information (i.e., send solicitations) to an initial set of users.

\textbf{Step 2):} Each solicited user decides whether to participate in the campaign: He first decides whether to consider a possible participation according to a \emph{participating interest factor}. If he does consider it and supposes to make a specific contribution following a \emph{contribution model}, he then evaluates the perceived reward according to a \emph{payout valuation model}. Finally, he decides to participate if his perceived reward outweighs the cost of participation following a \emph{cost model}.

\textbf{Step 3):} Each participant decides whether to solicit other users: He first predicts how many users from his acquaintances would accept his solicitations based on a \emph{solicitation prediction model}, and then computes the perceived gain from soliciting according to a \emph{payout valuation model}. Finally, he decides to send solicitations if his perceived gain outweighs the cost of sending solicitations following a \emph{cost model}. Each user's acquaintances are determined based on a \emph{social network model}.

\textbf{Step 4):} Repeat Steps 2) and 3) until the number of participants reaches the crowdsourcer's requirement or the campaign deadline arrives.

The aforementioned simulation framework is similar to \cite{douceur2007lottery}, which involves in a set of theories and models that have been widely accepted in literature. We briefly describe these models and some parameter settings as follows.

\textbf{Social Network Model:} An evolving network model \cite{toivonen2006model} is used to model the acquaintanceship of users, which exhibits several recognized properties of a social network, such as short average path length, broad degree distribution, high clustering, and community structure.
Three basic parameters for the model, $N_0$, $m_r$, and $m_s$, are set the same as specified by Toivonen et al. \cite{toivonen2006model}.

\textbf{Participating Interest Factor:} Each solicited user has two behavioral intentions: showing absolutely no interest or having an interest to consider whether to participate in. We assume each user has a \emph{participating interest factor}, $PIF$, to express his likelihood of two behavioral intentions.

\textbf{Contribution Model:}
%As described in Section \ref{sec:introduction}, both homogeneous and heterogeneous user models are considered, where users' contributions are assumed to follow a constant uniform distribution and a random uniform distribution, respectively.
As described in Section \ref{sec:introduction}, a more general model, \emph{heterogeneous user model}, is considered.  Specifically, each user $u$'s contribution, $C(u)$, is assumed to follow a random uniform distribution.

\textbf{Payout Valuation Model:} As described before, we leverage CPT to compute the perceived reward for different mechanisms. For Step 2), we first compute the lottery value based on the current tree structure, and then derive the perceived reward as described in Section \ref{subsec:mechanism selection}. For Step 3), we compute the perceived reward from soliciting new participants, and derive the difference between it and the original reward without sending solicitations as the perceived gain. The key parameters $\beta$ and $\delta$ are set the same as \cite{douceur2007lottery} for computing lottery values, and $\alpha$ and $\gamma$ are set the same as \cite{tversky1992advances} for leveraging CPT.

\textbf{Solicitation Prediction Model:} Each user $u$ assumes that all his $\zeta_u$ neighbors have not joined in the campaign, and each of his neighbors will join in the campaign with the probability $PIF$ if he sending solicitations. Thus, user $u$ will predict the number of users who accept his solicitations as $\zeta_u*PIF$.

\textbf{Cost Model:} Each user $u$ has a cost of participation, $CP(u)$, and another cost of sending solicitations, $CS(u)$, to represent his expected rewards, which are assumed to follow two different random uniform distributions.

The above models involve in many parameters as listed in Table \ref{table-parameter-settings}. Moreover, in order to evaluate the impacts of the budget constraint ($B$) and the number of required participants ($N$), we vary the values of $N$ from 5 to 50 with the increment of 1, and set two different values of $B$ as 1000 and 5000. For each setting, simulations are repeated 100 times and the respective average results are obtained, so as to reduce variance.
\begin{table}[htbp]\setlength{\tabcolsep}{5pt}
\begin{center}
\caption{Parameter Settings}
\label{table-parameter-settings}
\begin{tabular}{c|c}
  \hline
  Model & Parameter and Value\\
  \hline
  \hline
     & $N_0=30$,\\
   Social Network Model & $Pr(m_r\!=\!1)\!=\!0.95$, $Pr(m_r\!=\!0)\!=\!0.05$,\\
    & $m_s \sim U[1,3]$\\
   \hline
    Participating Interest Factor & $PIF = 0.5$\\
   \hline
   Contribution Model & $C(u)\sim U[1, 500]$\\
   \hline
    & $\beta=0.5$, $\delta=0.08$,\\
   \raisebox{0.9ex}[0pt]{Payout Valuation Model} & $\alpha=0.88$, $\gamma=0.61$\\
   \hline
    & $CP(u)\sim U[1,30]$\\
   \raisebox{0.9ex}[0pt]{Cost Model} & $CS(u)\sim U[1,15]$\\
  \hline
\end{tabular}
\end{center}
\end{table}

\subsection{Simulation Results}
\label{subsec:simulation results}
Fig. \ref{fig_simulation-results} shows the relationship between the number of required solicitations and the number of required participants for three lottree mechanisms, \emph{1-Pachira}, \emph{10-Pachira}, and \emph{Sharing-Pachira}, under different budget constraints. If we set the same budget constraint and the same number of required number of participants, then the less solicitations a mechanism requires, the easier it is to achieve the requirement of the crowdsourcing campaign. In other words, we use the number of required solicitations as a key metric for mechanism selection. First, when the budget stays the same, we can observe two common and interesting phenomenons independently from Fig. \ref{fig_reward1000random} and Fig. \ref{fig_reward5000random}:

i) The number of required solicitations increase with the number of required participants. %非常理所当然的。
Meanwhile, there is a larger and larger increasing rate of the number of required solicitations for \emph{Sharing-Pachira}, whereas \emph{1-Pachira} and \emph{10-Pachira} have a relatively lower increasing rate. It means that: as the number of required participants increases, \emph{Sharing-Pachira} will be harder and harder to achieve the requirement of the crowdsourcing campaign, whereas at this time, \emph{1-Pachira} and \emph{10-Pachira} could be the better choice.

ii) When a small number of participants is required, the number of required solicitations for the three lottree mechanisms presents the following relationships: \emph{Sharing-Pachira} $<$ \emph{10-Pachira} $<$ \emph{1-Pachira}, meaning that \emph{Sharing-Pachira} is the beast choice; when a large number of participants is required, it presents an opposite relationships: \emph{1-Pachira} $<$ \emph{10-Pachira} $<$ \emph{Sharing-Pachira}, meaning that \emph{1-Pachira} is the best choice; whereas \emph{10-Pachira} is almost always not the best choice. Generally, there is a distinct critical value of the required number of participants, below which one may prefer \emph{1-Pachira} to \emph{Sharing-Pachira}, and above which one may prefer \emph{Sharing-Pachira} to \emph{1-Pachira}. Specifically, this critical value is 10 (16) when $B=1000$ (5000).

Second, we can observe another interesting phenomenon by combining Fig. \ref{fig_reward1000random} and Fig. \ref{fig_reward5000random}:

iii) As the budget increases, the critical value of the required number of participants will become larger.

\textbf{Summary:} In essence, the above results are consistent with the CPT and the analysis in Sec. \ref{subsec:mechanism selection}, which also validate our important theoretical guidance to the mechanism selection in Sec. \ref{subsec:mechanism selection}.

\begin{figure}[!t]
  \centering{
  \subfigure[$B=1000$]{
    \label{fig_reward1000random}
    \includegraphics[width=3.0in]{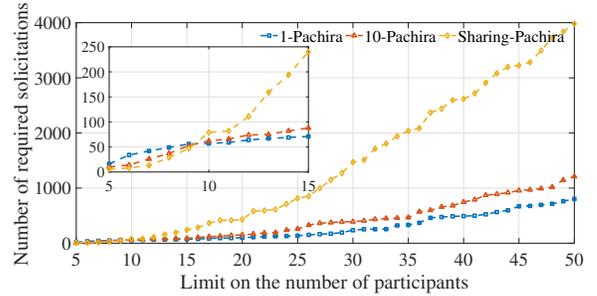}\vspace{-15pt}}
    \subfigure[$B=5000$]{
    \label{fig_reward5000random}
    \includegraphics[width=3.0in]{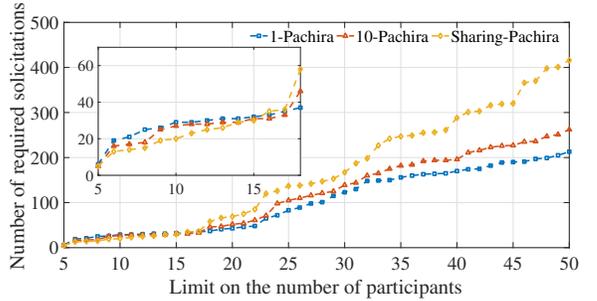}\vspace{-15pt}}
  }
  \caption{Relationship between the number of required solicitations and the number of required participants.}
  \label{fig_simulation-results}
  \vspace{-10pt}
\end{figure} 

\section{Looking For Lost Objects: An Application Case and Its Performance Evaluation}
%\section{Looking For a Lost Child: An Application Case and Its Actual Performance Evaluation}
\label{sec:experiments}
In this section, we first investigate an application case: \emph{looking for lost objects}. Then we design an experimental mobile game, \emph{Treasure Hunt}, and present a series of experiments and several metrics for evaluating the performance of lottree mechanisms. Finally, we provide the experimental results.

%有些人担心规则太复杂，难以操作
\subsection{Looking For Lost Objects: An Application Case}
\label{subsec:application case}
As elaborated in Section \ref{sec:introduction}, an incentive tree mechanism could be used in many crowdsourcing and mobile crowd sensing applications. Next we mainly consider a very useful application case: looking for lost objects, such as a lost child, pet, smartphone, key, and wallet. Imagine a child attaches a Bluetooth Low Energy (BLE) peripheral, e.g., Chipolo \cite{website:chipolo}, in his clothes or shoes. The low power consumption and miniaturization of BLE peripherals make it perfect for tracking the child continuously. If the child is lost, many smartphone users can be recruited to cooperatively look for him by continuous Bluetooth scanning and even locate him \cite{liu2014finding}. There is no doubt that incentive is a key to the success of this application.

\subsection{Treasure Hunt: An Experimental Mobile Game}
\label{subsec:treasure hunt}
In order to evaluate the performance of lottree mechanisms, we design an interesting experimental mobile game, \emph{Treasure Hunt}, which could be used directly for looking for lost objects due to the same intrinsic mechanism. The game involves in three roles:
\begin{itemize*}
\item
 a \emph{crowdsourcer} residing in the cloud, who is responsible for publicizing \emph{treasure hunt tasks}, monitoring users' participation process, and allocating rewards,
\item
 a set of \emph{users}, who register in our mobile APP to play games by using a Bluetooth-enabled smartphone, and
\item
 a \emph{treasure}, which is in fact a volunteer moving freely with a Bluetooth-enabled smartphone.
\end{itemize*}

In essence, \emph{Treasure Hunt} is to find the so-called \emph{treasure} by discovering its Bluetooth when a user gets close to it. A \emph{treasure hunt task} could be characterized by the reward budget, the number of required participants, the treasure ID (i.e., its Bluetooth ID), the task deadline, and the incentive type. Next we will introduce the operation procedure of the game, the contribution function design, and the incentive mechanism design, respectively.

\textbf{Operation Procedure of \emph{Treasure Hunt}:} It consists of six main points as follows.

i) The crowdsourcer publicizes a \emph{treasure hunt task}, and pushes the related information to all registered users. Note that only a part of users who run our mobile APP in the background and maintain a Internet connection can get the task information timely.

ii) Each user who received the task information decides whether to participate in. If yes, he turns on his Bluetooth, optionally opens his GPS, and periodically reports the participation information (time duration of Bluetooth scanning, GPS points) to the crowdsourcer.

iii) Each participating user decides whether to send solicitations to other users through a social network. The solicitation structure is recorded by the crowdsourcer.

iv) If some user discovers the treasure, then he reports the result to the crowdsourcer.

v) When the number of participants achieves the requirement, then no user can participate in any longer.

vi) After the deadline arrives, the crowdsourcer builds an incentive tree according to the participants' contributions and solicitation relationships, and then allocate rewards according to the announced incentive mechanism.

\textbf{Contribution Function Design of \emph{Treasure Hunt}:}
Intuitively, we hope that each user $u$ has a long time duration for Bluetooth scanning and a long travelling distance so as to more easily find the treasure. Thus, we design a contribution function, $C(u)$, by comprehensively considering three factors: time duration for Bluetooth scanning, $Dur(u)$, traveling distance, $Dis(u)$, and whether the user find the treasure, $Find(u)$ (a boolean function), namely that,
\begin{equation}\label{eq-contribution-design}
C(u)=0.5*Dur(u)+0.5*0.1*Dis(u)+120*Find(u).
\end{equation}
Here, $Dur(u)$ is measured in minutes, $Dis(u)$ is measured in meters, the number 0.5 is a weight factor, 0.1 is set because of that a traveling distance of 1 m takes about 0.1 min on average in our experiments, and 120 means that an extra contribution during 120 mins (the duration of a task) will be given to the user who finds the treasure.

\textbf{Incentive Mechanism Design of \emph{Treasure Hunt}:}
One of the most important objectives of \emph{Treasure Hunt} is to compare the three lottree mechanisms: 1-Pachira, $K$-Pachira, and Sharing-Pachira, by realistic experiments. However, it seems hard for users to understand the details of these mechanisms if we describe them straightforwardly. In fact, it is completely unnecessary for users to know about such complicated design. Instead, we only need to tell users a simple rule:

\emph{``Each user will earn a value, and will be rewarded based his value."}

To make it more intuitionistic for users to understand how to evaluate their values, we present the following descriptions to them:

\emph{``How to get a higher value: the longer duration you turn on your Bluetooth, the longer distance you travel (based on your GPS trajectory), the more friends you recommend to, then the higher value you get. Besides, the first participant and the participant who find the treasure will be given an extra value."}

Moreover, we show users intuitionistic descriptions on \emph{1-Pachira}, \emph{$K$-Pachira}, and \emph{Sharing-Pachira lottree} mechanisms, respectively:

\emph{``\textbf{Mechanism A:} Only one participant can get all the reward ($B$). Of course, the higher value, the more likely you win."}

\emph{``\textbf{Mechanism B:} We will have lottery drawings $K$ times, and each winner will get one $K$-th of the total reward ($B$). The higher value, the more likely you win."}

\emph{``\textbf{Mechanism C:} Every participant can get a reward. The higher value, the higher reward. But the total budget is $B$."}

\subsection{Experimental Settings and Evaluation Metrics}
\label{subsec:experimental settings}
We conduct \emph{Treasure Hunt} experiments in a university campus. For the convenience of comparing the three kinds of incentive mechanisms, we need to have a set of registered users as a basis. Thus, we first post an advertisement on our university BBS and some social groups in social networks (\emph{QQ} and \emph{WeChat}) seeking people to register in our mobile APP two days before the experiments officially start. In our advertisement, we tell users the game rules, and announce a budget of 500 RMB to recruit users, who will share the reward equally without limitation on the number of participants. Finally, 62 users registered in our APP before the experiments officially start. After that, we publicize 12 \emph{Treasure Hunt} tasks in 9 days. Each task begins at a random time and lasts for 2 hours. In order to investigate the impact of the number of required participants ($N$), we set two values of $N$ as 10 (tasks 1-3) and 50 (tasks 4, 6, 8) while fixing the budget constraint as $B=100$ RMB. Meanwhile, in order to investigate the impact of the budget constraint ($B$), we set three values of $B$ as 50 RMB (tasks 5, 7, 9), 100 RMB (tasks 4, 6, 8), and 500 RMB (tasks 10-12) while fixing the value of $N$ as $N=50$. The detailed settings are shown in Table \ref{table-experimental-settings}.

\begin{table}[htbp]\setlength{\tabcolsep}{5pt}
\begin{center}
\caption{Experimental Settings}
\label{table-experimental-settings}
\begin{tabular}{c|c|c|c|c}
  \hline
  Task & Release & Budget & Limit on \# of & Incentive\\
   No. & Date & ($B$) & participants ($N$) & Mechanism\\
  \hline
  \hline
   1 & Jan. 5, 2018 & 100 RMB & 10 & \emph{1-Pachira}\\
   2 & Jan. 6, 2018 & 100 RMB & 10 & \emph{Sharing-Pachira}\\
   3 & Jan. 7, 2018 & 100 RMB & 10 & \emph{5-Pachira}\\
   4 & Jan. 8, 2018 & 100 RMB & 50 & \emph{1-Pachira}\\
   5 & Jan. 8, 2018 & 50 RMB  & 50 & \emph{1-Pachira}\\
   6 & Jan. 9, 2018 & 100 RMB & 50 & \emph{Sharing-Pachira}\\
   7 & Jan. 9, 2018 & 50 RMB  & 50 & \emph{Sharing-Pachira}\\
   8 & Jan. 10, 2018 & 100 RMB & 50 & \emph{5-Pachira}\\
   9 & Jan. 10, 2018 & 50 RMB  & 50 & \emph{5-Pachira}\\
   10 & Jan. 11, 2018 & 500 RMB & 50 & \emph{1-Pachira}\\
   11 & Jan. 12, 2018 & 500 RMB & 50 & \emph{Sharing-Pachira}\\
   12 & Jan. 13, 2018 & 500 RMB & 50 & \emph{5-Pachira}\\
  \hline
\end{tabular}
\end{center}
\end{table}

Generally, three performance metrics should be concerned: \emph{total number of participants}, \emph{total contribution of participants}, and \emph{average contribution of participants}. However, two practical factors need to be considered. First, it is a common phenomenon that users' participation enthusiasm declines over time, which has been described in some literature \cite{lee2010sell,gao2015providing}, and verified through a long-term experiment \cite{ji2017exploring}. It means that it is not fair to directly compare the \emph{total number of participants}, as our experiments span a long time. In order to reduce the effect of this factor, we consider another metric, \emph{total number of active users}, meaning the number of users who have ever opened the APP in a certain day. Note that, the reason that an active user opened the APP may be his interest in the APP itself or in the certain task. Whereas a participator must be interested in the certain task. Thus, we use a metric called \emph{Relative Participation Ratio (RPR)} to represent the actual attractiveness of a task, defined as follows:
\begin{equation}\label{eq-participation-ratio}
\texttt{RPR} = \frac{\texttt{\# of participants}}{\texttt{\# of active users} - \texttt{\# of participants}},
\end{equation}
where the denominator could be used to indicate the actual activeness of users that is independent of a certain task.

Second, there is a strong randomness on whether a user can find the treasure. In order to reduce the effect of this factor on evaluating different incentive mechanisms, we revise the contribution function in Eq. (\ref{eq-contribution-design}) as follows:
\begin{equation}\label{eq-contribution-design2}
C(u)=0.5*Dur(u)+0.5*0.1*Dis(u),
\end{equation}
which is used for computing the \emph{total contribution of participants (TCP)} and \emph{average contribution of participants (ACP)}.

\subsection{Experimental Results}
\label{subsec:experimental results}
In our experiments, 20 new users registered in our APP, resulting in 82 registered users in total by adding 62 initial users. However, there are always some inactive users each day. First, we verify the phenomenon that users' participation enthusiasm declines over time. Fig. \ref{fig_activeUsers} shows the changes in the number of active users in 9 days. Generally, there is a significant decline in the number of active users over time. Although the number of active users shows a transient increase on Jan. 8 and Jan. 11, one big reason is the increase of the budget ($B$) or the limit on the number of participants ($N$). Moreover, the number of active users shows a significant decreasing trend over time for the same settings of $B$ and $N$ (by comparing Jan. 5-7, Jan. 8-10, and Jan. 11-13, respectively). This justifies the usage of the metric \emph{RPR} as explained earlier.

Next, we analyze the experimental results on the three metrics introduced earlier: RPR, TCP, and ACP. Moreover, the impacts of the budget constraint ($B$) and the number of required participants ($N$) are investigated. Note that, when we set $N=10$, the number of participants achieves the limitation for all of the three incentive mechanisms. Thus, it is unnecessary to consider the RPR and ACP for the experiments in Jan. 5-7.

\textbf{Relative Participation Ratio (RPR):} Fig. \ref{fig_RPR} shows the RPR under different budget constraints when we fix $N=50$. When $B=50$ RMB, the 1-Pachira lottree has the highest RPR; when $B=100$ RMB, three mechanisms' RPRs are very close; when $B=500$ RMB, the Sharing-Pachira lottree has the highest RPR.

\textbf{Total Contribution of Participants (TCP):} Fig. \ref{fig_con1} plots the TCP under different values of $N$ when we fix $B=100$, from which we observe that the Sharing-Pachira lottree has the highest TCP when a small number of participants is required, while the 1-Pachira lottree has the highest TCP when a large number of participants is required. Fig. \ref{fig_con2} plots the TCP under different values of $B$ when we fix $N=50$, from which we observe that the 1-Pachira lottree has the highest TCP when there is a small budget constraint (50 RMB and 100 RMB), while the Sharing-Pachira lottree has the highest TCP when there is a large budget constraint (500 RMB).

\textbf{Average Contribution of Participants (ACP):} Fig. \ref{fig_avgCon} plots the ACP under different budget constraints when we fix $N=50$. When $B=50$ RMB, the 1-Pachira lottree has a slightly higher ACP than the other two mechanisms; when $B=100$ RMB, the 1-Pachira lottree has a similar ACP as the 5-Pachira lottree, which is clearly higher than the Sharing-Pachira lottree; when $B=500$ RMB, the Sharing-Pachira lottree has the highest ACP, which is clearly higher than the other two mechanisms.

\textbf{Summary:} In essence, the above results are almost all consistent with the CPT and the analysis in Sec. \ref{subsec:mechanism selection}, which also validate our important theoretical guidance to the mechanism selection in Sec. \ref{subsec:mechanism selection}. Note that, some results seem not to be very matched with the theoretical analysis or our intuition. For example, the budget 100 RMB results in lower RPR, TCP, and RPR than the budget 50 RMB when we fix $N=50$. It might be due to the impact of task release time or order. For another, the 5-Pachira lottree is sometimes best but sometimes worst. It exhibits slight instability, the reason of which is difficult, if not impossible, to understand as human psychology and behavior are themselves very complex. Nevertheless, it does not affect the obvious regularity from our experimental results that is indeed very matched with our theoretical guidance in Sec. \ref{subsec:mechanism selection}.

\begin{figure}[!t]
  \centering{
    \includegraphics[width=3.2in]{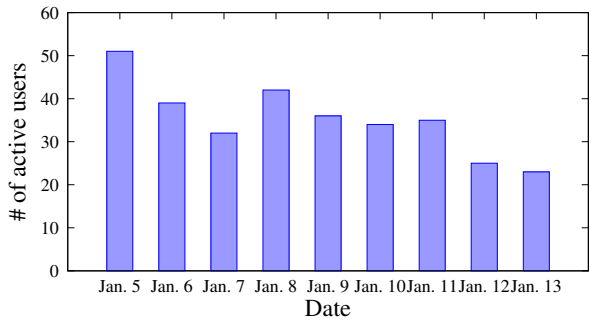}
  }
  \caption{Changes in the number of active users in 9 days.}
  \label{fig_activeUsers}
  \vspace{-10pt}
\end{figure}

\begin{figure}[!t]
  \centering{
    \includegraphics[width=2.5in]{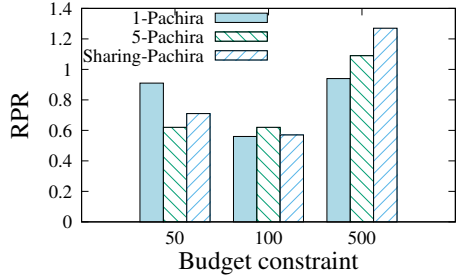}
  }
  \caption{Relative participation ratio under different budget constraints when the limit on the number of participants is fixed as $N=50$.}
  \label{fig_RPR}
  \vspace{-10pt}
\end{figure}

\begin{figure}[!t]
  \centering{
    \includegraphics[width=2.5in]{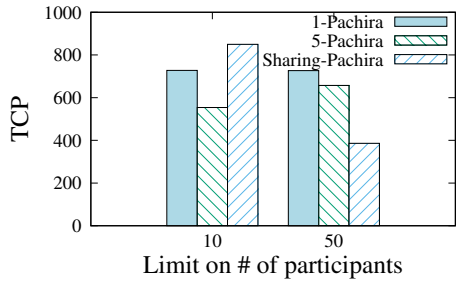}
  }
  \caption{Total contribution of participants under different limits on the number of participants when the budget is fixed as $B=100$ RMB.}
  \label{fig_con1}
  \vspace{-10pt}
\end{figure}

\begin{figure}[!t]
  \centering{
    \includegraphics[width=2.5in]{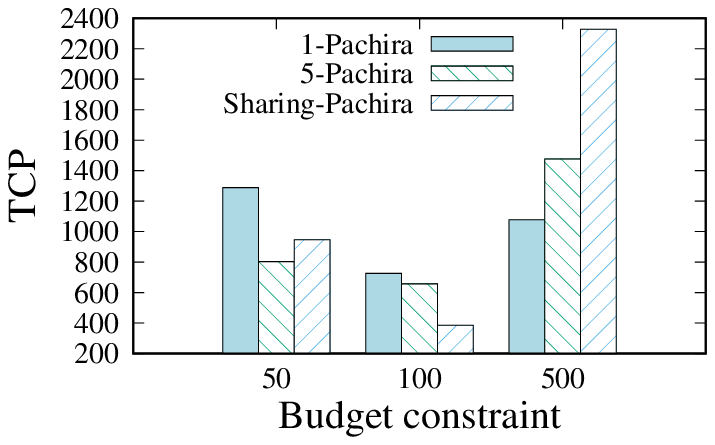}
  }
  \caption{Total contribution of participants under different budget constraints when the limit on the number of participants is fixed as $N=50$.}
  \label{fig_con2}
  \vspace{-10pt}
\end{figure}

\begin{figure}[!t]
  \centering{
    \includegraphics[width=2.5in]{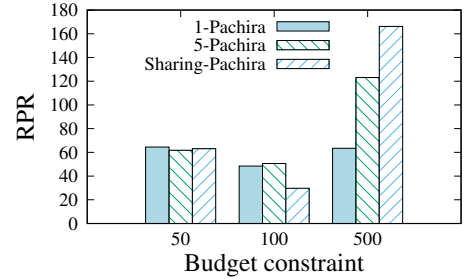}
  }
  \caption{Average contribution of participants under different budget constraints when the limit on the number of participants is fixed as $N=50$.}
  \label{fig_avgCon}
  \vspace{-10pt}
\end{figure} 

\section{Related Work}
\label{sec:related work}
Many incentive mechanisms have been proposed for crowdsourcing \cite{yang2012crowdsourcing,duan2012incentive,lee2010sell,jaimes2012location,zhang2015incentivize,zhang2015truthful,zhao2016budget,zhao2017frugal,guo2017taskme}. Generally, these mechanisms fall into two categories: the crowdsourcer-centric mechanisms, where the crowdsourcer provides a fixed reward to participates, and the user-centric mechanisms, where users have their reserve prices for crowdsourcing services. For the former, a Stackelberg game is often used by assuming that the costs of participants or their probability distribution is known \cite{yang2012crowdsourcing,duan2012incentive}. For the latter, various types of auctions are often used \cite{lee2010sell,jaimes2012location,zhang2015incentivize,zhang2015truthful,zhao2016budget,zhao2017frugal,guo2017taskme}. Lee and Hoh \cite{lee2010sell} designed a dynamic auction mechanism for purchasing users' sensing data. Jaimes et al. \cite{jaimes2012location} further considered a budget constraint and users' locations. Yang et al. \cite{yang2012crowdsourcing} proposed the \emph{MSensing} auction mechanism, and proved that it satisfied several properties including computational efficiency, individual rationality, profitability, and truthfulness. Zhang et al. \cite{zhang2015incentivize} proposed an auction mechanism for incentivizing crowd workers to label a set of binary tasks under a strict budget constraint. Zhang et al. \cite{zhang2015truthful} considered three auction models, which involve cooperation and competition among users. Zhao et al. proposed two kinds of online auction mechanisms: budget-feasible mechanisms \cite{zhao2016budget} and frugal mechanisms \cite{zhao2017frugal}. Guo et al. \cite{guo2017taskme} proposed a dynamic and quality-enhanced auction mechanism.

On the other hand, some incentive tree mechanisms have been investigated in various fields. Emek et al. \cite{emek2011mechanisms} presented multi-level marketing mechanisms that motivate participants to promote a certain product among their friends through social networks. Drucker and Fleischer \cite{drucker2012simpler} proposed a family of multi-level marketing mechanisms that preserve natural properties and are simple to implement. Chen et al. \cite{chen2013sybil} designed efficient sybil-proof incentive mechanisms, called the direct referral mechanisms, for retrieving information from networked agents. Zhang et al. \cite{zhang2015sybil} proposed a sybil-proof incentive tree mechanism for crowdsourcing scenarios where the contribution model is considered to be submodular and time-sensitive. Lv and Moscibroda \cite{lv2016fair} presented two families of incentive tree mechanisms for crowdsourcing, where each family achieves a set of desirable properties. Zhang et al. \cite{zhang2017robust} designed an auction-based incentive tree mechanism for mobile crowd sensing which combines the advantages of auctions and incentive trees. However, all of these studies failed to account for a budget constraint. To the best of our knowledge, only the early work \cite{douceur2007lottery} designed a class of incentive tree mechanisms with budget constraint, but they violate BC and allow only one winner.

Besides, some studies have been conducted to examine incentive mechanisms by real-world experiments. Reddy et al. \cite{reddy2010examining} examined various micro-payment schemes from a pilot study in a university campus sustainability initiative. Musthag et al. \cite{musthag2011exploring} et al. used a combination of statistical analysis and models from
labor economics to evaluate three micro-payment schemes in the context of high-burden user studies. Celis et al. \cite{celis2013lottery} investigated the benefits and potential pitfalls in employing a lottery-based payment mechanism for crowdsourcing via experiments on \emph{MTurk}. Rula et al. \cite{rula2014no} compared micro-payments and lottery based schemes by using data from a large, 2-day experiment with 96 participants at a corporate conference. Rokicki et al. \cite{rokicki2014competitive} compared three classes of reward schemes, \emph{linear reward}, \emph{competitive-based}, and \emph{lottery-based} by large-scale experimental evaluations. They further investigated how team mechanisms can be leveraged to improve the cost efficiency of crowdsourcing \cite{rokicki2015groupsourcing}. However, all of these studies lacked a general and solid theoretical basis to account for their experimental results, and none of them considered incentive tree mechanisms.
%However, none of them considered incentive tree mechanisms.

\section{Conclusion}
\label{sec:conclusion}
In this paper, we investigated budget-consistent incentive tree mechanisms for crowdsourcing. We proposed three types of \emph{generalized lottree} mechanisms, \emph{1-Pachira}, \emph{$K$-Pachira}, and \emph{Sharing-Pachira} for allowing one winner, multiple winners, and each participant to be a winner, respectively. We proved that our mechanisms satisfy BC, CCI, CSI, VPC, USB and USA. A theoretical guidance to the mechanism selection was provided for satisfying different requirements. Both extensive simulations and realistic experiments were conducted to confirm our theoretical analysis.

% references section
\bibliographystyle{IEEEtran}
\bibliography{myRef}

\end{document}